\documentclass{amsart}
\usepackage{amsmath,enumerate,enumitem,amsfonts,color,amsthm,hyperref,bbm} 

\usepackage{appendix}

\title[Exponential Hedging for the Ornstein-Uhlenbeck Process]{Exponential Hedging for the Ornstein-Uhlenbeck Process in the Presence of Linear Price Impact}

\author{Yan Dolinksy} 
\address{Department of Statistics, Hebrew University}
\email{yan.dolinsky@mail.huji.ac.il}

\date{\today}

\numberwithin{equation}{section}  
\newtheorem{defn}{Definition}[section]

\newtheorem{remark}[defn]{Remark}
\newtheorem{thm}[defn]{Theorem}

\newtheorem{lem}[defn]{Lemma}

\begin{document}

\begin{abstract}
In this work we study a continuous time exponential utility maximization problem in the presence of a linear temporary price impact.
More precisely, for the case where the risky asset is given by the Ornstein-Uhlenbeck diffusion process we compute the optimal portfolio strategy and the corresponding value. 
Our method of solution relies on duality, and it is purely probabilistic.
\end{abstract}

\keywords{Exponential Utility Maximization, Linear Price Impact, Ornstein-Uhlenbeck Process.}

\maketitle

\section{Introduction and the Main Result}\label{sec:1}

We fix a time horizon $T<\infty$ and consider a financial market which consists of a safe asset earning zero interest rate and of a risky asset $S_t$, $t\in [0,T]$ that follows 
the Ornstein-Uhlenbeck diffusion process 
\begin{equation}\label{model}
dS_t=\left(\mu-S_t\right) dt+dW_t,  
\end{equation}
where $\mu\in\mathbb R$ is the long-term mean and 
$W_t$, $t\geq 0$ is a 
standard, one-dimensional Brownian motion defined on a probability space $(\Omega,\mathcal F,\mathbb P)$ 
and endowed with the augmented natural filtration $\mathcal F_t$, $t\geq 0$. As usual, we assume that the initial stock price is a given constant $S_0$.
Without loss of generality we assume that $\mathcal F=\mathcal F_T$.

Our motivation to study optimal investments with a mean reverting risky asset comes from the recent paper \cite{GNR}
where 
the authors considered asymptotically (as the time horizon goes to infinity) optimal investments 
for mean reverting risky assets.
Although the authors considered a more general type of mean reverting models then the linear 
Ornstein-Uhlenbeck diffusion, they assumed that the trading is frictionless.  

In financial markets, trading moves prices against the trader:
buying faster increases execution prices, and selling faster decreases them.
This aspect of liquidity, known as market depth
\cite{B} or price-impact has recently received increasing attention
(see, for instance,
\cite{AFS:2010,BCE:2021,CHM:20,FSU:19,MMS:17,N:20} and the references therein).

Following
\cite{AlmgrenChriss:01}, we model this impact in a temporary linear
form and, thus, when at time $t$ the investor turns over her position $\Phi_t$ at
the rate $\phi_t=\dot{\Phi}_t$ the execution price is $S_t+\frac{\phi_t}{2\delta}$
for some constant $\delta>0$ which denotes the market depth.
 As a result, the profits and
losses from trading are given by
\begin{equation}\label{por}
V^{\Phi_0,\phi}_T:=\Phi_0(S_T-S_0)+\int_{0}^T \phi_t(S_T-S_t)dt-\frac{1}{2\delta} \int_{0}^T \phi^2_t dt,
\end{equation}
where, for convenience, we assume that the investor marks to market her position
$\Phi_T = \Phi_0+\int_0^T \phi_t dt$ in the risky asset that she has acquired by time
$T>0$.

The natural class of admissible strategies
is given by 
\begin{align*}
\mathcal A:=\left\{\phi=(\phi_t)_{t\in [0,T]}: \ \phi \text{ is } \
  \mathcal F_t\text{-optional with } \int_{0}^T \phi^2_t dt<\infty
  \ \text{ a.s.}\right \}.
\end{align*}

Let us explain formula (\ref{por}) in more detail. At time $0$, the investor holds $\Phi_0$ shares of the stock and an amount $-\Phi_0 S_0$ in the savings account.
At time $t \in [0,T]$, the investor purchases $\phi_t\,dt$, an infinitesimal number of shares. Consequently, the (infinitesimal) change in the savings account is
$
-\phi_t\left(S_t+\frac{\phi_t}{2\delta}\right)\,dt .
$
The terminal portfolio value consists of the terminal balance in the savings account together with the value of the stock holdings. Hence it is given by
$$-\Phi_0 S_0-\int_0^T \phi_t\left(S_t+\frac{\phi_t}{2\delta}\right)\,dt+\Phi_T S_T.$$
Since
$
\Phi_T=\Phi_0+\int_{0}^T \phi_t\,dt,
$
we obtain the right-hand side of (\ref{por}).

The investor’s preferences are described by an exponential utility function
$u(x)=-\exp(-\alpha x)$, $x\in\mathbb R$,
 with absolute risk aversion parameter $\alpha>0$, and her goal is to
\begin{align}\label{problem}
\text{Maximize } \mathbb E_{\mathbb P}\left[u(V^{\Phi_0,\phi}_T)\right]=\mathbb
  E_{\mathbb P}\left[-\exp\left(-\alpha V^{\Phi_0,\phi}_T\right)\right]\text{ over
  } {\phi\in\mathcal A},
\end{align}
where $\mathbb E_{\mathbb P}$ denotes the expectation with respect to the underlying probability measure $\mathbb P$. 
Observe that for any $(\Phi_0,\phi)\in \mathbb{R}\times\mathcal{A}$, we have
$
\alpha V_T^{\Phi_0,\phi}=\tilde V_T^{\alpha\Phi_0,\alpha\phi},
$
where $\tilde V^{\cdot}_{\cdot}$ is defined by (\ref{por}) with market depth $\tilde\delta:=\alpha\delta$. Thus, without loss of generality, we 
take $\alpha=1$.

The paper's main result is the following solution to the optimization
problem (\ref{problem}). For simplicity we assume that the initial number of shares is $\Phi_0=0$.
\begin{thm}\label{thm1.1}
Set
$$
\mathcal V(t):=\frac{\sqrt{ 1+\delta} \left(\sinh\left(\sqrt{1+\delta} t \right) + \sqrt{1+\delta} \left(1+\delta+\delta t\right) \cosh\left(\sqrt{1+\delta} \, t\right)\right)}{\sqrt{1+\delta} \left(1+\delta+\delta t\right) \sinh\left(\sqrt{1+\delta} \, t\right) + \left(1+\delta^2\right) \cosh\left(\sqrt{1+\delta} \, t\right) + 2{\delta}}
-1, \ \ t\geq 0.
$$
The utility maximization problem~\eqref{problem} has a unique optimal 
hedging strategy $\hat \phi_t$, $t\in [0,T]$  which is given by 
the feedback form
\begin{equation}\label{port}
\hat\phi_t=\delta\frac{\left(1+\frac{1-\delta}{(1+\delta)^{3/2}}
\tanh\left(\frac{\sqrt{1+\delta}(T-t)}{2}\right)+\frac {\delta (T-t) }{1+\delta}
\right)(\mu-S_t)-\hat\Phi_t}{1+\frac {\delta (T-t) }{1+\delta}+\frac{\sqrt{1+\delta}}{\sinh(\sqrt{1+\delta} (T-t))}+\frac{1+\delta^2}{(1+\delta)^{3/2}}\tanh\left(\frac{\sqrt{1+\delta}(T-t)}{2}\right)}
\end{equation}
where $\hat\Phi_t:=\int_{0}^t\hat\phi_s ds$, $t\in [0,T].$
The corresponding value is equal to 
\begin{equation}\label{value}
\mathbb
  E_{\mathbb P}\left[-\exp\left(- V^{0,\hat\phi}_T\right)\right]=-\exp\left(-\frac{1}{2}\mathcal V(T)(\mu-S_0)^2-\frac{1}{2}\int_{0}^T \mathcal V(t)dt\right).
  \end{equation}
  \end{thm}

Let us collect some financial-economic observations from this result.
First, direct calculations yield that the derivative of $\mathcal V$ is equal to
\begin{align*}
\dot{\mathcal V}(t)=\frac{\delta(1+\delta)\left(A(t)+B(t)+\delta^2\sinh^2\left(\sqrt{1+\delta}t\right)\right) }{\left(\sqrt{1+\delta} \left(1+\delta+\delta t\right) \sinh\left(\sqrt{1+\delta} \, t\right) + \left(1+\delta^2\right) \cosh\left(\sqrt{1+\delta} \, t\right) + 2{\delta}\right)^{2}}
\end{align*}
where
\begin{align*}
&A(t):=2(1+\delta+\delta t)\sqrt{1+\delta}\sinh\left(\sqrt{1+\delta}t\right)-2(1+\delta)^2 t\\
&B(t):=2(1+\delta)\cosh\left(\sqrt{1+\delta}t\right)-2(1+\delta)-(\delta+\delta^2) t^2.
\end{align*}
Clearly, $\sinh\left(\sqrt{1+\delta}t\right)\geq\sqrt{1+\delta}t$ and
$\cosh\left(\sqrt{1+\delta}t\right)\geq 1+\frac{1+\delta}{2} t^2$ for all $t$. Thus, $A,B\geq 0$.
We conclude that $\dot{\mathcal V}(t)>0$ for $t>0$. 
Since $\mathcal V(0)=0$ we get $\mathcal V(t)>0$ for $t>0$. 
Hence,
from (\ref{value}) the optimal expected utility is increasing in the following two parameters:
I. Time Horizon $T$. II. The term $|\mu-S_0|$, namely the distance between the initial stock price and the long-term mean of the 
Ornstein-Uhlenbeck diffusion process. 

Obviously, for $(\Phi_0,\phi)\in \mathbb{R}\times\mathcal{A}$, the corresponding portfolio is increasing in the market depth $\delta$. Consequently, the value of the optimization problem~\eqref{problem} is increasing in the market depth~$\delta$. As $\delta \downarrow 0$, we have $\mathcal{V}(t) \downarrow 0$, and hence the corresponding expected utility converges to $-1$. This is intuitively clear: as the market depth approaches zero, any nonzero trading strategy incurs an infinite penalty, making the zero strategy optimal.

On the other hand, as the market depth tends to infinity, the function $\mathcal{V}(t)$ converges to~$t$. Consequently, the value converges to
$
-e^{-\frac{1}{2}T|\mu - S_0|^2 - \frac{1}{2}T^2},
$
which coincides with the value of the frictionless utility maximization problem (see Section~3 in~\cite{GNR}).

Next, for a fixed $\delta>0$, we have $\lim_{t\rightarrow\infty} \mathcal V(t)=\sqrt{1+\delta}-1$.
Thus, the certainty equivalent 
$c(T):= -\log\left(\mathbb E_{\mathbb P}\left[\exp\left(- V^{0,\hat\phi}_T\right)\right]\right)$ (the optimal strategy $\hat\phi$ depends on the time horizon $T$)
satisfies 
\begin{equation}\label{limit}
\lim_{T\rightarrow\infty} \frac{c(T)}{T}=\frac{\sqrt{1+\delta}-1}{2}.
\end{equation}
Let us notice that for the frictionless case (see Section 3 in \cite{GNR}) the certainty equivalent is of order $T^2$ (for large $T$).

The feedback 
description (\ref{port}) says that the optimal trading strategy 
is a mean reverting strategy towards the process 
\begin{equation*}
X_t:=\left(1+\frac{1-\delta}{(1+\delta)^{3/2}}
\tanh\left(\frac{\sqrt{1+\delta}(T-t)}{2}\right)+\frac {\delta (T-t) }{1+\delta}
\right)(\mu-S_t), \ \ t\in [0,T].
\end{equation*}
This process depends linearly on the difference between the long-term mean $\mu$ and the current stock price $S_t$. 
If the current time $t$ is close to the maturity date $T$, the value of $X_t$ is close $\mu-S_t$. 
If $T-t$ is large then $X_t$ is of order $\frac{\delta(T-t)}{1+\delta}(\mu-S_t)$.
The denominator in (\ref{port}) describes
how fast (as a function of time) the investor should trade towards the process $X$. Notice that the  denominator is large if the time left until the maturity date $T-t$ is large or small. 
Namely, in both of these cases the investor should trade slowly towards the process $X$.

For the case where the market depth goes to infinity, i.e. $\delta\rightarrow\infty$, the process $X$
converges (almost surely) to $(1+T-t)(\mu-S_t)$, $t\in [0,T]$. 
Observe that the denominator in (\ref{port}) is of order $\sqrt \delta$ (for large $\delta$) and so, for the case where 
the friction vanishes the optimal number of shares $\hat\Phi_t$ 
converges to $(1+T-t)(\mu-S_t)$ which is exactly the optimal strategy for 
the frictionless utility maximization problem (see Section 3 in \cite{GNR}). 

We end this section with the following two remarks.

\begin{remark}
For a simplified presentation and economy of notation, we assume that the mean-reversion rate and the volatility of the Ornstein--Uhlenbeck diffusion process are normalized to $1$. Observe that for given $\kappa,\sigma>0$, the process
$
Y_t := \frac{\sigma}{\sqrt{\kappa}} S_{\kappa t}$, $t\geq 0$
satisfies the SDE
\[
dY_t = \kappa\left(\frac{\sigma \mu}{\sqrt{\kappa}} - Y_t\right) dt + \sigma d\hat W_t,
\qquad 
Y_0 = \frac{\sigma S_0}{\sqrt{\kappa}},
\]
where $\hat W$ is a standard Brownian motion defined by
$
\hat W_t := \frac{1}{\sqrt{\kappa}} W_{\kappa t}$, $t\geq 0$.
Moreover, for any $(\Phi_0,\phi)\in \mathbb{R}\times\mathcal{A}$ we have
\[
V^{\Phi_0,\phi}_T
=
\hat\Phi_0\left(Y_{\hat T} - Y_0\right)
+ \int_{0}^{\hat T} \hat\phi_{t} \left(Y_{\hat T} - Y_t\right)\, dt
- \frac{1}{2\hat\delta} \int_{0}^{\hat T} \hat\phi_t^2 \, dt,
\]
where
$\hat{\Phi}_0 := \frac{\Phi_0 \sqrt{\kappa}}{\sigma}$, 
$ \hat T := \frac{T}{\kappa}$, $\hat\delta := \frac{\kappa^2\delta}{\sigma^2}$
and
$\hat\phi_t := \frac{\kappa^{3/2}}{\sigma}\phi_{\kappa t}$, $t\in [0,\hat T]$.

Thus, in the case where the initial number of shares is $\Phi_0 = 0$, the value of the exponential utility maximization problem with parameters--maturity date $\hat T > 0$, initial stock price $\hat S_0$, long-term mean $\hat\mu$, mean-reversion $\kappa > 0$, volatility $\sigma > 0$, and market depth $\hat\delta > 0$--is equal to the value of (\ref{problem}) with the rescaled parameters: maturity date $\kappa \hat T$, initial stock price $\frac{\sqrt{\kappa}}{\sigma}\hat S_0$, long-term mean $\frac{\sqrt{\kappa}}{\sigma}\hat\mu$, and market depth $\frac{\sigma^2}{\kappa^2}\hat\delta$.

From (\ref{limit}) and the above analysis, we obtain that the corresponding certainty equivalent satisfies
$
\lim_{T \to \infty} \frac{c(T)}{T}
=  \frac{\sqrt{1 + \frac{\sigma^2 \hat\delta}{\kappa^2}} - 1}{2}.
$
Observe that this expression is decreasing in the mean-reversion parameter $\kappa > 0$ and increasing in the volatility $\sigma > 0$.

On the other hand, suppose that the maturity date is fixed and the market depth goes to infinity. Then, from the discussion following Theorem \ref{thm1.1}, we conclude that the corresponding value converges (as $\hat\delta \to \infty$) to
$
- e^{-\frac{\kappa^2}{\sigma^2} T |\hat\mu - \hat S_0|^2 - \frac{1}{2}\kappa^2 \hat T^2}.
$
This expression is clearly increasing in $\kappa$ and decreasing in $\sigma$.
Thus, there is no uniform pattern regarding monotonicity with respect to the mean-reversion and volatility parameters.
\end{remark}

\begin{remark}
The proof of Theorem \ref{thm1.1} relies on the duality result for exponential hedging with quadratic transaction costs given in Proposition A.2 of \cite{BDR:22}. From a financial perspective, one limitation of our setup is that we do not impose liquidation of the position at the maturity date, i.e., we do not include the constraint $\Phi_T = 0$.

A duality result for the case with a liquidation constraint is also available (see \cite{D:24}). In particular, Theorem 2.1 in \cite{D:24} extends Proposition A.2 of \cite{BDR:22} to the setting where liquidation at maturity is required. The liquidation constraint introduces additional difficulties in the dual representation. Specifically, the maximization is taken over all equivalent probability measures together with the corresponding martingales, whereas in Proposition A.2 of \cite{BDR:22} the dual variables consist only of equivalent probability measures (compare (2.3)--(2.4) in \cite{D:24} with (\ref{duality})--(\ref{portfolio}) in the current paper).

This additional structure makes the resulting computations substantially more involved. In \cite{D:24}, for example, we considered the problem of hedging a quadratic option in the Bachelier model. Even in this relatively simple setting, the computations arising from the dual formulation are quite complicated (see Proposition 4.1 in \cite{D:24}). In fact, we were only able to determine the optimal portfolio, but not the corresponding value of the optimization problem.

For these reasons, we do not impose a liquidation constraint in the present paper. Moreover, even in the current setting without a liquidation constraint, in order to express the value of the optimization problem \eqref{problem} in a tractable form, we assume that the initial number of shares satisfies $\Phi_0 = 0$. This assumption simplifies the computations in Lemma \ref{lem3}.
\end{remark}

\section{Proof of Theorem~\ref{thm1.1}}\label{sec:2}
Our proof will be based on a change of measure and the duality result given by Proposition A.2 in \cite{BDR:22}.

Set $B_t := S_t - S_0$, $t \ge 0$. 
From (\ref{model}),
$$
B_t=(S_0 - \mu)\left(e^{-t}-1\right)+\int_0^t e^{-(t-s)}\, dW_s, \ \ t\geq 0.$$
Namely, $B_t$, $t\in [0,T]$, is of the form
$
B_t=\gamma_t+\int_{0}^t m_{t,s} \, dW_s,
$
for $\gamma\in L^2[0,T]$ and a Volterra kernel $m\in L^2([0,T]^2)$.
Since
$dB_t = dW_t + (\mu - S_0 - B_t)\,dt$, $t\geq 0$
then from Theorem 2' in \cite{HH:68} it follows that there exists a probability measure 
$\hat{\mathbb P}\sim\mathbb P$ such that $B_t$, $t\in [0,T]$ is a Brownian motion with respect to $\hat{\mathbb P}$.
Clearly,
$
dW_t = dB_t - (\mu - S_0 - B_t)\,dt$, $t\geq 0$.
Thus, by applying Theorem 2' in \cite{HH:68} again, we obtain that the corresponding Radon--Nikodym derivative is given by
$
\frac{d\mathbb P}{d\hat{\mathbb P}} = e^X,
$
where
\begin{align}\label{1}
X
&:= \int_{0}^T (\mu - S_0 - B_t)\,dB_t
   - \frac{1}{2}\int_{0}^T (\mu - S_0 - B_t)^2\, dt \nonumber\\
&= \frac{T}{2}
  + (\mu - S_0)B_T
  - \frac{B_T^2}{2}
  - \frac{1}{2}\int_{0}^T (\mu - S_0 - B_t)^2\, dt .
\end{align}

Next, we apply Proposition A.2 in \cite{BDR:22}. We begin by verifying the integrability condition in Assumption A.1 of \cite{BDR:22}.
Since $B_t$, $t\in [0,T]$ is a Brownian motion with respect to $\hat{\mathbb P}$, there exists a constant $c>0$ such that
$
\mathbb E_{\hat{\mathbb P}}\!\left[\exp\!\left(c\max_{0\le t\le T} B_t^2\right)\right] < \infty .
$
This, together with the Cauchy--Schwarz inequality, (\ref{1}) and the elementary bound
$
S_t^2 \le 2(S_0^2 + B_t^2)$, $t\geq 0$
gives
\[
\mathbb{E}_{\mathbb{P}}\!\left[\exp\!\left(\frac{c}{4}\max_{0\le t\le T} S_t^2\right)\right]
\le 
\mathbb{E}_{\hat{\mathbb{P}}}\!\left[
\exp\!\left(\frac{cS_0^2+T}{2}+(\mu-S_0)B_T+\frac{c}{2}\max_{0\le t\le T} B_t^2\right)
\right]
<\infty .
\]
Hence Assumption A.1 in \cite{BDR:22} holds.

Clearly, 
for any probability measure
$\mathbb Q\sim\mathbb P$ we have
$\log\left(\frac{d\mathbb Q}{d{\mathbb P}}\right)=
\log\left(\frac{d\mathbb Q}{d\hat{\mathbb P}}\right)-X$.
Thus, from Proposition A.2 in \cite{BDR:22} we obtain the following duality 
\begin{align}\label{duality}
&\max_{\phi\in\mathcal A}\left\{-\log\left(\mathbb E_{\mathbb P}\left[\exp\left(-V^{\Phi_0,\phi}_T\right)\right]\right)\right\}\nonumber\\
&=\inf_{\mathbb Q\in \mathcal Q}\mathbb E_{\mathbb Q}\left[\Phi_0 B_T-X+\log\left(\frac{d\mathbb Q}{d\hat{\mathbb P}}\right)+\frac{\delta}{2}\int_{0}^T\left(
\mathbb E_{\mathbb Q}[B_T-B_t|\mathcal F_t]\right)^2 dt\right]
\end{align}
where $\mathcal Q$ is the set of all probability
measures $\mathbb Q\sim \hat{\mathbb P}$ with finite relative entropy
$\mathbb E_{\mathbb Q}\left[\log\frac{d\mathbb Q}{d\hat{\mathbb P}} \right]<\infty$. 
    Said proposition also yields that there exists a unique
solution $\hat{\mathbb Q}$ to the dual problem,  and it allows us to construct the unique 
solution $\hat{\phi}$ to the primal problem as
\begin{equation}\label{portfolio}
  \hat{\phi}_t = \delta  \mathbb E_{\hat{\mathbb Q}}[B_T-B_t|\mathcal F_t], \quad
  t \in [0,T].
\end{equation}
We arrive at the following lemma. 
\begin{lem}\label{lem1}
The dual infimum in (\ref{duality}) coincides with the one taken over
  all $\mathbb Q \in \mathcal Q$ whose densities take the form
  \begin{equation}\label{2}
\frac{d\mathbb Q}{d\hat{\mathbb P}} = \exp\left(\int_0^T \eta_t 
    dB_t-\frac{1}{2}\int_0^T\eta_t^2dt\right) 
  \end{equation}
  for some bounded and adapted $\eta$ changing values only at finitely
  many deterministic times. For such $\mathbb Q$ the induced value for
  the dual problem can be written as
  \begin{align}\label{equality}
&\mathbb E_{\mathbb Q}\left[\Phi_0 B_T-X+\log\left(\frac{d\mathbb Q}{d\hat{\mathbb P}}\right)+\frac{\delta}{2}\int_{0}^T\left(
\mathbb E_{\mathbb Q}[B_T-B_t|\mathcal F_t]\right)^2 dt\right]\nonumber\\
&=-\frac{T}{2}+\left(\Phi_0-(\mu-S_0)\right)\int_{0}^T a_t dt+\frac{1}{2}\left(\int_{0}^T a_t dt\right)^2+\frac{1}{2}\int_{0}^T a^2_t dt\\
&+\frac{1}{2}\int_{0}^T 
\left(\int_{0}^t a_s ds-(\mu-S_0)\right)^2 dt+\frac{\delta}{2}\int_{0}^T \left(\int_{t}^T a_s ds\right)^2 dt\nonumber\\
&+\int_{0}^T \mathbb E_{\mathbb Q}\left[\frac{1}{2}\left(1+\int_{s}^T l_{t,s} dt\right)^2+\frac{1}{2}\int_{s}^T\left(1+\int_{s}^t l_{u,s}du\right)^2 dt \right.\nonumber\\
&\left.+\frac{1}{2}\int_{s}^T l^2_{t,s}dt+\frac{\delta}{2}\int_{s}^T \left(\int_{t}^T l_{u,s} du\right)^2 dt\right]ds\nonumber
\end{align}
  where, for $t\in[0,T]$, $a_t$ and $l_{t,.}=(l_{t,s})_{s\in [0,t]}$ such that $\mathbb E_{\mathbb Q}\left[\int_{0}^t l^2_{t,s}ds\right]<\infty$  are
  determined by the It\^o-representations
  \begin{equation}\label{3}
     \eta_t = a_t + \int_0^t l_{t,s} dW^{\mathbb Q}_s
  \end{equation}
  with respect to the $\mathbb Q$-Brownian motion $W^{\mathbb
    Q}_s=B_s-\int_0^s \eta_r dr$, $s\in [0,T]$.
\end{lem}
\begin{proof}
 For any $\mathbb Q \in \mathcal Q$ the martingale representation property of Brownian motion gives us a predictable $\eta$ with $\mathbb E_{\mathbb Q}[\log(d\mathbb Q/d\hat{\mathbb P})]=\mathbb E_{\mathbb Q}[\int_0^T\eta^2_sds]/2<\infty$ such that the density $d\mathbb Q/d\hat{\mathbb P}$ takes the form (\ref{2}). Using this density to rewrite the dual target functional as an expectation under $\hat{\mathbb P}$, we can follow standard density arguments in $L^2(dt\otimes\hat{\mathbb P})$ (see Section 3.2 in \cite{KS}) to see that the infimum over $\mathbb Q \in \mathcal Q$ can be realized by considering the $\mathbb Q$ induced via (\ref{2}) by simple $\eta$ as described in the lemma's formulation. As a consequence, the It\^o representations of $\eta_t$ in (\ref{3}) can be chosen in such a way that the resulting $(a_t,l_{t,.})$ are also measurable in $t$: in fact they only change when $\eta$ changes its value, i.e., at finitely many deterministic times. This joint measurability will allow us below to freely apply Fubini’s theorem. Let us rewrite the dual target functional in terms of $a$ and $l$. In terms of $\eta$ and the $\mathbb Q$-Brownian motion $W^{\mathbb Q}$.
  
First, from It\^o's isometry and Fubini's theorem we obtain
\begin{equation}\label{4}
\mathbb E_{\mathbb Q}\left[\int_{0}^T \eta^2_u du\right]=\int_{0}^T a^2_t dt+\int_{0}^T\int_{s}^{T}\mathbb E_{\mathbb Q}\left[l^2_{t,s} \right]dt\, ds.
\end{equation}
Again by Fubini's theorem it follows that
\begin{align*}
&\mathbb E_{\mathbb Q}\left[\int_{t}^T \eta_udu \middle|\mathcal F_t\right]=\int_{t}^T a_u du+\mathbb E_{\mathbb Q}\left[\int_{0}^T\int_{t}^T l_{u,s}du\, dW^{\mathbb Q}_s\middle|\mathcal F_t\right]\\
&=\int_{t}^T a_u du+\int_{0}^{t} \int_{t}^T l_{u,s} du\, dW^{\mathbb Q}_s, \ \ t\in [0,T]
\end{align*}
where the last equality follows from the
martingale property of stochastic integrals. 
Thus, another application of It\^o's isometry and Fubini's theorem yields
\begin{align}\label{5}
&\int_{0}^T\left(\mathbb E_{\mathbb Q}[B_T-B_t|\mathcal F_t]\right)^2 dt=\int_{0}^T\left(\mathbb E_{\mathbb Q}\left[\int_{t}^T \eta_udu \middle|\mathcal F_t\right]\right)^2 dt\nonumber\\
&=\int_{0}^T\left(\int_{t}^T a_u du\right)^2 dt+\int_{0}^T\mathbb E_{\mathbb Q}\left[
\int_{s}^T \left(\int_{t}^T l_{u,s} du\right)^2 dt\right]ds.
\end{align}
By applying again the It\^o's isometry and the Fubini theorem 
\begin{equation}\label{6}
\mathbb E_{\mathbb Q}\left[B^2_T\right]=\left(\int_{0}^T a_t dt\right)^2+\int_{0}^T\mathbb E_{\mathbb Q}\left[\left(1+\int_{s}^T l_{u,s} du\right)^2\right] ds
\end{equation}
and 
\begin{align}\label{7}
&\mathbb E_{\mathbb Q}\left[\int_{0}^T \left(B_t-(\mu-S_0)\right)^2 dt\right]\nonumber\\
&=\int_{0}^T 
\left(\int_{0}^t a_s ds-(\mu-S_0)\right)^2 dt+\int_{0}^T \mathbb E_{\mathbb Q}\left[\int_{s}^T\left(1+\int_{s}^t l_{u,s} du\right)^2dt\right]ds.
\end{align}
Plugging (\ref{1}) and (\ref{4})--(\ref{7}) into the left hand side of (\ref{equality}) provides us with the claimed formula.
\end{proof}

We now have all the pieces in place that we need for the 
\textbf{Completion of the proof of Theorem  \ref{thm1.1}}.
\begin{proof}
The crucial point of the above representation is that by taking the
minimum separately over $a$ and over $l_{.,s}$ for each $s \in [0,T]$
we obtain deterministic variational problems that can be solved
explicitly and this
deterministic minimum yields the solution to the dual problem. 
Namely, 
from Lemma \ref{lem3} it follows that there exists deterministic functions $\hat a,\hat l$ which minimize the right hand side of (\ref{equality}).
Theorem 2 in \cite{HH:68} yields a probability measure $\hat{\mathbb Q}$ which satisfies
$$
\frac{d\hat{\mathbb Q}}{d\hat{\mathbb P}} = \exp\left(\int_0^T \hat\eta_t 
    dB_t-\frac{1}{2}\int_0^T\hat\eta_t^2dt\right) 
  $$
 for $\hat\eta_t = \hat a_t + \int_0^t \hat l_{t,s} dW^{\hat{\mathbb Q}}_s$, $t\in [0,T]$. The probability measure $\hat{\mathbb Q}$ is the unique solution 
 to the dual problem given by the right-hand side of (\ref{duality}). 
 
From (\ref{portfolio}) and (\ref{property}) ($\hat h=\hat a$, $s=0$) we obtain 
 $$
 \hat\phi_0=\delta\int_{0}^T \hat a_t dt=\frac{\left(1+\frac{1-\delta}{(1+\delta)^{3/2}}
\tanh\left(\frac{\sqrt{1+\delta} T}{2}\right)+\frac {\delta T }{1+\delta}
\right)(\mu-S_t)-\Phi_0}{1+\frac {\delta T }{1+\delta}+\frac{\sqrt{1+\delta}}{\sinh(\sqrt{1+\delta}T)}+\frac{1+\delta^2}{(1+\delta)^{3/2}}\tanh\left(\frac{\sqrt{1+\delta} T}{2}\right)}.
$$
This gives (\ref{port}) for $t=0$. 
We claim that (\ref{port}) holds for every \(t \in [0,T]\). To this end, we apply the dynamic programming argument used in Lemma 3.5 of \cite{BDR:22}.
Assume, by contradiction, that the statement fails. By (\ref{portfolio}), the sample paths of the unique optimal portfolio \(\hat{\phi}\) are continuous. Hence, there exists \(t_0 \in [0,T]\) such that, with positive probability, \(\hat{\phi}_{t_0}\) does not coincide with the right-hand side of (\ref{port}). Now consider the strategy \(\tilde{\phi}\), which coincides with \(\hat{\phi}\) up to time \(t_0\) and is given by (\ref{port}) for \(t\geq t_0\).
For any strategy $\phi$, we can write the contribution over the interval $[t_0,T]$ to the resulting
terminal wealth as
\[
V_T^{\Phi_0,\phi} - V_{t_0}^{\Phi_0,\phi}
= {\Phi}_{t_0}(S_T-S_{t_0})
+ \int_{t_0}^T \phi_t(S_T-S_t)\,dt
- \frac{1}{2\delta}\int_{t_0}^T \phi_t^2\,dt
=: V^{\Phi_{t_0},\phi}_{[t_0,T]},
\]
where, \(V_{t_0}^{\Phi_0,\phi}\) denotes the portfolio value at time \(t_0\), defined by (\ref{por}) with \(T\) replaced by \(t_0\).
Thus, by the Markov property of \(S_t\), \(t\in[0,T]\), and by the choice of \(\tilde{\phi}\) as the unique optimal policy from time \(t_0\) onward, we obtain
\[
\mathbb{E}_{\mathbb{P}}\!\left[
-\exp\!\left(- V^{\tilde\Phi_{t_0},\tilde\phi}_{[t_0,T]}\right)
\,\middle|\,\mathcal F_{t_0}
\right]
\ge
\mathbb{E}_{\mathbb{P}}\!\left[
-\exp\!\left(-V^{\hat\Phi_{t_0},\hat\phi}_{[t_0,T]}\right)
\,\middle|\,\mathcal F_{t_0}
\right],
\]
with strict inequality on \(\{\hat\phi_{t_0}\neq \tilde\phi_{t_0}\}\) (i.e., where (\ref{port}) is violated), since the continuity of \(\hat\phi\) and \(\tilde\phi\) implies that once they differ at \(t_0\), they must differ on an open interval.

Since, by assumption, this occurs with positive probability, it follows for the unconditional expected utility of \(\tilde\phi\) that
\begin{align*}
\mathbb{E}_{\mathbb{P}}\bigl[-\exp(-V_T^{\Phi_0,\tilde\phi})\bigr]
&=
\mathbb{E}_{\mathbb{P}}\!\left[
\exp(-V_{t_0}^{\Phi_0,\tilde\phi})
\mathbb{E}_{\mathbb{P}}\!\left[
-\exp\!\left(-V^{\tilde\Phi_{t_0},\tilde\phi}_{[t_0,T]}\right)
\,\middle|\,\mathcal F_{t_0}
\right]
\right]
\\
&>
\mathbb{E}_{\mathbb{P}}\!\left[
\exp(- V_{t_0}^{\Phi_0,\tilde\phi})
\mathbb{E}_{\mathbb{P}}\!\left[
-\exp\!\left(-V^{\hat\Phi_{t_0},\hat\phi}_{[t_0,T]}\right)
\,\middle|\,\mathcal F_{t_0}
\right]
\right]
\\
&=
\mathbb{E}_{\mathbb{P}}\bigl[-\exp(- V_T^{\Phi_0,\hat\phi})\bigr],
\end{align*}
contradicting the optimality of $\hat\phi$.

Next, we establish (\ref{value}). By applying (\ref{number}) for $\theta:=-1$ we get 
\begin{align*}
&\frac{1}{2}\left(1+\int_{s}^T \hat l_{t,s} dt\right)^2+\frac{1}{2}\int_{s}^T\left(1+\int_{s}^t \hat l_{u,s}du\right)^2 dt\\
&+\frac{1}{2}\int_{s}^T \hat l^2_{t,s}dt+\frac{\delta}{2}\int_{s}^T \left(\int_{t}^T \hat l_{u,s} du\right)^2 dt=\frac{1}{2}\left(\mathcal V(T-s)+1\right), \ \ s\in [0,T].
\end{align*}
By applying (\ref{number}) for $\theta:=\mu-S_0$ we get 
\begin{align*}
&\left(\Phi_0-(\mu-S_0)\right)\int_{0}^T \hat a_t dt+\frac{1}{2}\left(\int_{0}^T \hat a_t dt\right)^2+\frac{1}{2}\int_{0}^T \hat a^2_t dt\\
&+\frac{1}{2}\int_{0}^T 
\left(\int_{0}^t \hat a_s ds-(\mu-S_0)\right)^2 dt+\frac{\delta}{2}\int_{0}^T \left(\int_{t}^T\hat a_s ds\right)^2 dt=\mathcal V(T)(\mu-S_0)^2.
\end{align*}
By combining the above two equalities with (\ref{duality}) and (\ref{equality})
we obtain (\ref{value}) and complete the proof.  
\end{proof}

\section{Auxiliary Results}\label{sec:3}
We start with the following simple result. Since we could not find a direct reference we provide the full details. 
\begin{lem}\label{lem2}
Let $\alpha>0$. For any $s\in [0,T)$ and $x,y\in\mathbb R$ let $\Gamma^s_{x,y}$ be the space 
of all absolutely continuous functions $g:[s,T]\rightarrow \mathbb R$
which satisfy $g(s)=x$ and $g(T)=y$.
Then,
\begin{align*}
&\min_{g\in\Gamma^s_{x,y}}\left[\frac{1}{2}\int_{s}^T\dot{g}_t^2 dt+\frac{1}{2}\alpha^2\int_{s}^T g^2_t dt\right]\\
&=\frac{1}{2}\alpha\left(\frac{(x-y)^2}{\sinh(\alpha (T-s))}+\tanh(\alpha (T-s)/2)(x^2+y^2)\right).\nonumber
\end{align*}
\end{lem}
\begin{proof}
The left-hand side of the above equation can be written as the optimization problem
$\min_{g \in \Gamma^s_{x,y}} \int_{s}^T H(\dot g_t, g_t)\,dt$,
where $H(u,v) := \tfrac{1}{2}u^2 + \tfrac{1}{2}\alpha^2 v^2$ for $u,v\in\mathbb R$.

Since $H$ is a positive definite quadratic form, the above optimization problem admits a unique minimizer. This minimizer necessarily satisfies the Euler--Lagrange equation (see, for example, Section 8.2 in \cite{E}, applied to the one-dimensional case).
Thus, the optimizer $\hat g$ is the unique solution to the ODE
$$\ddot{\hat g}(t)=\alpha^2 \hat g(t), \ \ \hat g(s)=x, \ \ \hat g(T)=y.$$
We conclude that 
$$\hat g(t)=\frac{x\sinh(\alpha(T-t))+y\sinh(\alpha (t-s))}{\sinh(\alpha (T-s))}, \ \ t\in [s,T].$$
Hence, simple computations yield 
\begin{align*}
&\min_{g\in\Gamma^s_{x,y}}\left[\frac{1}{2}\int_{s}^T\dot{g}_t^2 dt+\frac{1}{2}\alpha^2\int_{s}^T g^2_t dt\right]=\frac{1}{2}\int_{s}^T\dot{\hat g}_t^2 dt+\frac{1}{2}\alpha^2\int_{s}^T \hat g^2_t dt\\
&=\frac{\alpha^2}{2\sinh^2(\alpha (T-s))}\int_{s}^T \left(x^2\cosh(2\alpha (T-t))+y^2 \cosh (2\alpha(t-s))\right)dt\\
&-\frac{\alpha^2 x y}{2\sinh^2(\alpha (T-s))}  \int_{s}^T \cosh(\alpha (2t-s-T))) dt\\
&=\frac{1}{2}\alpha \left(\coth(\alpha (T-s))(x^2+y^2)-\frac{2 x y}{\sinh(\alpha (T-s))}\right)\\
&=\frac{\alpha}{2}\left(\frac{(x-y)^2}{\sinh(\alpha (T-s))}+\tanh(\alpha (T-s)/2))(x^2+y^2)\right).
\end{align*}
\end{proof}

Next, we apply Lemma \ref{lem2} in order to derive the following computations. 
\begin{lem}\label{lem3}
Let $\theta\in \mathbb R$ be a constant 
and $s \in [0,T)$. Consider the optimization problem 
\begin{align}\label{op}
&\min_{h\in L^2[s,T]}\left\{(\Phi_0-\theta)\int_{s}^T h_t dt+\frac{1}{2}\left(\int_{s}^T h_t dt\right)^2+\frac{1}{2}\int_{s}^T h^2_t dt\right.\nonumber\\
&\left.+\frac{1}{2} \int_{s}^T 
\left(\int_{s}^t h_v dv-\theta\right)^2 dt+\frac{\delta}{2}\int_{s}^T \left(\int_{t}^T h_v dv\right)^2 dt\right\}.
\end{align}
The above optimization problem has a unique optimizer $\hat h$ which satisfies 
\begin{equation}\label{property}
\int_{s}^T \hat h_t dt=\frac{\left(1+\frac{1-\delta}{(1+\delta)^{3/2}}
\tanh\left(\frac{\sqrt{1+\delta}(T-s)}{2}\right)+\frac {\delta (T-s) }{1+\delta}
\right)\theta-\Phi_0}{1+\frac {\delta (T-s) }{1+\delta}+\frac{\sqrt{1+\delta}}{\sinh\left(\sqrt{1+\delta} (T-s)\right)}+\frac{1+\delta^2}{(1+\delta)^{3/2}}\tanh\left(\frac{\sqrt{1+\delta} (T-s)}{2}\right)}.
\end{equation}
Moreover, for the case $\Phi_0=0$ we have 
\begin{align}\label{number}
&\min_{h\in L^2[s,T]}\left\{-\theta\int_{s}^T h_t dt+\frac{1}{2}\left(\int_{s}^T h_t dt\right)^2+\frac{1}{2}\int_{s}^T h^2_t dt\right.\nonumber\\
&\left.+\frac{1}{2} \int_{s}^T 
\left(\int_{s}^t h_v dv-\theta\right)^2 dt+\frac{\delta}{2}\int_{s}^T \left(\int_{t}^T h_v dv\right)^2 dt\right\}=\frac{1}{2}\mathcal V(T-s)\theta^2.
\end{align}
\end{lem}
   \begin{proof}
   The uniqueness follows from the strict convexity of the functional. 
 For any $z\in\mathbb R$ let $\Delta^s_z\subset L^2[s,T]$ be the set of all functions $h\in L^2[s,T]$
    which satisfy $\int_{s}^T h_{t} dt-\theta =z$. 
Let $z>0$ and let $h\in \Delta^s_z$. 
    The function $f(t):=\int_{s}^t h_u du-\theta$, $t\in [s,T]$ satisfies
$\dot{f}=h$ a.e., and so,
    \begin{align*}\label{equality}
&(\Phi_0-\theta)\int_{s}^T h_t dt+\frac{1}{2}\left(\int_{s}^T h_t dt\right)^2+\frac{1}{2}\int_{s}^T h^2_t dt\nonumber\\
&+\frac{1}{2} \int_{s}^T 
\left(\int_{s}^t h_v dv-\theta\right)^2 dt+\frac{\delta}{2}\int_{s}^T \left(\int_{t}^T h_v dv\right)^2 dt\\
&=(\Phi_0-\theta)(z+\theta)+\frac{1}{2}(z+\theta)^2+\frac{1}{2}\int_{s}^T \dot{f}^2(t)dt\\
&+\frac{1}{2}\int_{s}^T f^2(t) dt+\frac{\delta}{2} \int_{s}^T (z-f(t))^2 dt\\
&=(\Phi_0-\theta)(z+\theta)+\frac{1}{2}(z+\theta)^2+\frac{1}{2}\frac {\delta (T-s)}{1+\delta}z^2\\
&+\frac{1}{2}\int_{s}^T \dot{f}^2(t)dt+\frac{1}{2}(1+\delta)\int_{s}^T \left(f(t)-\frac{\delta z}{1+\delta}\right)^2 dt.
\end{align*}
Next, by applying Lemma \ref{lem2} for $\alpha:=\sqrt{1+\delta}$, 
$g:=f-\frac{\delta z}{1+\delta}$, 
and observing 
the relations $g(s)=-\theta-\frac{\delta z}{1+\delta }$, $g(T)=\frac{z}{1+\delta}$,
 we get 
 \begin{align*}
&\min_{h\in \Delta^s_z}\left\{(\Phi_0-\theta)\int_{s}^T h_t dt+\frac{1}{2}\left(\int_{s}^T h_t dt\right)^2+\frac{1}{2}\int_{s}^T h^2_t dt\right.\nonumber\\
&\left.+\frac{1}{2} \int_{s}^T 
\left(\int_{s}^t h_v dv-\theta\right)^2 dt+\frac{\delta}{2}\int_{s}^T \left(\int_{t}^T h_v dv\right)^2 dt\right\}\\
&=(\Phi_0-\theta)(z+\theta)+\frac{1}{2}\ (z+\theta)^2+\frac{1}{2}\frac {\delta (T-s) }{1+\delta}z^2+\frac{1}{2}\frac{\sqrt{1+\delta}}{\sinh\left(\sqrt{1+\delta} (T-s)\right)}(z+\theta)^2\\
&+\frac{1}{2}\sqrt{1+\delta}\tanh\left(\frac{\sqrt{1+\delta} (T-s)}{2}\right)\left(\left(\theta+\frac{\delta z}{1+\delta}\right)^2+\left(\frac{ z}{1+\delta}\right)^2\right)\\
&=\Phi_0\theta+\frac{1}{2}\left(\sqrt{1+\delta}\coth\left((\sqrt{1+\delta} (T-s)\right)-1\right)\theta^2\\
&+\left(\Phi_0+\frac{\delta\theta}{\sqrt{1+\delta}}\coth\left(\sqrt{1+\delta} (T-s)\right)+\frac{\theta}{\sqrt{1+\delta}\sinh\left(\sqrt{1+\delta} (T-s)\right)}\right)z\\
&+\frac{1}{2}\left(1+\frac {\delta (T-s) }{1+\delta}+\frac{\sqrt{1+\delta}}{\sinh\left(\sqrt{1+\delta} (T-s)\right)}+\frac{1+\delta^2}{(1+\delta)^{3/2}}\tanh\left(\frac{\sqrt{1+\delta} (T-s)}{2}\right)\right)z^2.
\end{align*}

Since $L^2[s,T]=\bigcup_{z\in\mathbb R}\Delta^s_z$, then we are looking for a 
$z$ which minimizes the above quadratic form. 
The minimum of the above quadratic is attained at the unique point 
$$\hat z:=-\frac{\Phi_0+\frac{\delta\theta}{\sqrt{1+\delta}}\coth\left(\sqrt{1+\delta} (T-s)\right)+\frac{\theta}{\sqrt{1+\delta}\sinh\left(\sqrt{1+\delta} (T-s)\right)}}{1+\frac {\delta (T-s) }{1+\delta}+\frac{\sqrt{1+\delta}}{\sinh\left(\sqrt{1+\delta} (T-s)\right)}+\frac{1+\delta^2}{(1+\delta)^{3/2}}\tanh\left(\frac{\sqrt{1+\delta} (T-s)}{2}\right)}.$$
Thus, the unique solution to the optimization problem (\ref{op}) is an element in $\Delta^s_{\hat z}$. 
Simple calculations yield
$$
\hat z+\theta=\frac{\left(1+\frac{1-\delta}{(1+\delta)^{3/2}}
\tanh\left(\frac{\sqrt{1+\delta}(T-s)}{2}\right)+\frac {\delta (T-s) }{1+\delta}
\right)\theta-\Phi_0}{1+\frac {\delta (T-s) }{1+\delta}+\frac{\sqrt{1+\delta}}{\sinh\left(\sqrt{1+\delta} (T-s)\right)}+\frac{1+\delta^2}{(1+\delta)^{3/2}}\tanh\left(\frac{\sqrt{1+\delta} (T-s)}{2}\right)}
$$
and
(\ref{property}) follows. 

Finally,
for the case where $\Phi_0=0$, the minimal value of the above quadratic form (in $z$) is
\begin{align*}
&\frac{1}{2}\left(\sqrt{1+\delta}\coth\left(\sqrt{1+\delta} (T-s)\right)-1\right)\theta^2\\
&-\frac{1}{2}\frac{\left(\frac{\delta}{\sqrt{1+\delta}}\coth\left(\sqrt{1+\delta} (T-s)\right)+\frac{1}{\sqrt{1+\delta}\sinh\left(\sqrt{1+\delta} (T-s)\right)}\right)^2\theta^2}{1+\frac {\delta (T-s) }{1+\delta}+\frac{\sqrt{1+\delta}}{\sinh\left(\sqrt{1+\delta} (T-s)\right)}+\frac{1+\delta^2}{(1+\delta)^{3/2}}\tanh\left(\frac{\sqrt{1+\delta} (T-s)}{2}\right)}\\
&=\frac{1}{2}\left(\frac{\left(1+\frac {\delta (T-s) }{1+\delta}\right)\sqrt{1+\delta}\coth\left(\sqrt{1+\delta} (T-s)\right)+\frac{1}{1+\delta}}{{1+\frac {\delta (T-s) }{1+\delta}+\frac{\sqrt{1+\delta}}{\sinh\left(\sqrt{1+\delta} (T-s)\right)}+\frac{1+\delta^2}{(1+\delta)^{3/2}}\tanh\left(\frac{\sqrt{1+\delta} (T-s)}{2}\right)}}-1\right)\theta^2\\
&=\frac{1}{2}\mathcal V(T-s)\theta^2
\end{align*}
where the last two equalities follow from direct computations. 
\end{proof}

\section*{Acknowledgments}

I would like to thank the referees for their careful reading of the manuscript and for their insightful comments and suggestions, which have significantly improved the quality and clarity of the paper. 
Partially supported by the Israel Science Foundation (ISF) grant no 305/25.


\begin{thebibliography}{10}
\bibitem{AlmgrenChriss:01}
R. Almgren and N. Chriss,
{\em Optimal execution of portfolio transactions,}
Journal of Risk,
{\bf 3,} 5--39, (2001).

\bibitem{AFS:2010}
A. Alfonsi, A. Fruth and A. Schied,
{\em Optimal execution strategies in limit order books with general shape functions,}
Quantitative Finance,
{\bf 10,} 143--157, (2010).

\bibitem{B}
F. Black,
{\em Noise,}
Journal of Finance,
{\bf 41,} 529--543, (1986).

\bibitem{BCE:2021}
E. Bayraktar, T. Caye and I. Ekren,
{\em Asymptotics for Small Nonlinear Price Impact: a PDE Approach to the Multidimensional Case,}
Mathematical Finance,
{\bf 31,} 36--108, (2021).

\bibitem{BDR:22}
P. Bank, Y. Dolinsky and M. R\'{a}sonyi,
{\em What If We Knew What the Future Brings? Optimal Investment for a Frontrunner with Price Impact,}
Applied Mathematics and Optimization,
{\bf 86,} Article 25, (2022).

\bibitem{CHM:20}
T. Caye, M. Herdegen and J. Muhle-Karbe,
{\em Trading with Small Nonlinear Price Impact,}
Annals of Applied Probability,
{\bf 30,} 706--746, (2020).

\bibitem{D:24}
Y. Dolinsky,
{\em Duality Theory for Exponential Utility-Based Hedging in the Almgren--Chriss Model,}
Journal of Applied Probability,
{\bf 61,} 420--438, (2024).

\bibitem{E}
L. C. Evans,
{\em Partial Differential Equations (2nd ed.),}
American Mathematical Society,
(2010).

\bibitem{FSU:19}
A. Fruth, T. Schöneborn and M. Urusov,
{\em Optimal Trade Execution in Order Books with Stochastic Liquidity,}
Mathematical Finance,
{\bf 29,} 507--541, (2019).

\bibitem{GNR}
P. Guasoni, L. Nagy and M. R\'{a}sonyi,
{\em Young, Timid, and Risk Takers,}
Mathematical Finance,
{\bf 31,} 1332--1356, (2021).

\bibitem{HH:68}
M. Hitsuda,
{\em Representation of Gaussian Processes Equivalent to Wiener Process,}
Osaka Journal of Mathematics,
{\bf 5,} 299--312, (1968).

\bibitem{KS}
I. Karatzas and S. E. Shreve,
{\em Brownian Motion and Stochastic Calculus (2nd ed.),}
Springer,
(1991).

\bibitem{MMS:17}
L. Moreau, J. Muhle-Karbe and H. M. Soner,
{\em Trading with Small Price Impact,}
Mathematical Finance,
{\bf 27,} 350--400, (2017).

\bibitem{N:20}
S. Nadtochiy,
{\em A Simple Microstructural Explanation of the Concavity of Price Impact,}
Mathematical Finance,
{\bf 32,} 78--113, (2022).

\end{thebibliography}
\end{document}